\documentclass[conference,a4paper]{IEEEtran}

\usepackage{amsmath,amssymb,amsbsy,amsthm}

\let\epsilon\varepsilon
\let\phi\varphi

\def\X{{\cal X}}

\def\N{\mathbb N}

\def\S{\mathcal S}
\def\Sp{{\mathcal S}^+}

\def\E{\mathbb E}

 \newtheorem{theorem}{Theorem}

\theoremstyle{definition}

\title{Predicting the outcomes of  every process for which an asymptotically accurate stationary predictor exists is impossible}
 \author{\IEEEauthorblockN{Daniil Ryabko }
   \IEEEauthorblockA{INRIA Lille, France\\\texttt{\small daniil@ryabko.net}}	
\and
\IEEEauthorblockN{Boris Ryabko}
\IEEEauthorblockA{
Institute of Computational Technologies, \\ 
SB RAS, 
 Russia\\  \texttt{  boris@ryabko.net} }
}
\date{}
\begin{document}
\maketitle
\begin{abstract}
The problem of prediction consists in forecasting the conditional
distribution of the next outcome given the past.  Assume that  the source generating the data is such that there is a stationary
ergodic predictor whose error converges to zero (in a certain sense). The question is whether there is a universal 
 predictor for all such sources, that is, a predictor whose error goes to zero if any of the sources that have this property 
is chosen to generate the data. This question is answered in the negative, contrasting a number of previously 
established positive results concerning related but smaller sets of processes.
\end{abstract}

\section{Introduction}
The basic problem is predicting the conditional probability distribution 
$\hat \mu(\cdot|x_1,\dots,x_n)$ over the next outcome $x_{n+1}$ given a sequence of observations 
$x_1,\dots,x_n$ generated by an unknown time-series distribution $\mu$. Since $\hat \mu$ gives a conditional 
distribution for every $x_1,\dots,x_n$ it defines itself a time-series distribution. Thus, the source
of data and the predictor are objects of the same kind.   Traditionally, one assumes
$x_i$ to be independent and identically distributed, or that $\mu$ belongs to one of the  well-studied parametric
families. However, in applications involving hard-to-model  data sources such as stock market, human-written
texts or biological sources, it is often assumed, instead, that $\mu$ belongs to some large (nonparametric) family 
of time-series distributions. Examples of such families are the set of all finite-memory distributions or 
the set of all stationary distributions.   The hope is not that the unknown data source 
under study actually belongs to such a family~-- for example, that a human-written text obeys the finite-memory
assumption or that the stock market is stationary~-- but, rather, that the considered family of sources is good
enough for the forecasting task at hand. Such a ``hope,'' however, remains informal, since the theoretical
results concern the setting when the unknown source belongs to the family. 

Here we consider a formalization of the ``good enough for prediction'' setting proposed in \cite{Ryabko:11pq4+}. Specifically, 
we are asking  whether a predictor can be constructed which 
is asymptotically consistent (prediction error goes to 0) on any source for which a consistent predictor exists in a given family.  Thus, given a set $\S$ of distributions, 
we consider the set  $\Sp:=\{${\em of all distributions $\mu$ such that there exists
a distribution $\nu\in\S$ such that the prediction error of $\nu$ on sequences generated by $\mu$ goes to zero}$\}$.
We are asking whether there exists a predictor that is consistent on all distributions in $\Sp$. In this
work  the family $\S$ in question is the set of all stationary ergodic distributions, and the question is answered in the negative.
This negative result is rather tight; in particular, the same proof shows that the set $S$ 
can be replaced by the set of all Hidden Markov chains with a countable set of states (maintaining the negative result), 
while a consistent predictor exists if we only consider Hidden Markov chains with a finite set of states. 

\noindent{\bf Prior work.} A consistent predictor for the set $\S$ of all finite-alphabet stationary 
distributions has been constructed in \cite{BRyabko:88}. Here the prediction quality 
is with respect to Cesaro-averaged Kullback-Leibler (KL) divergence, which is required to 
converge to 0 either in  expectation or with probability~1. The same work shows that 
an analogous result is impossible to obtain without Cesaro averaging (the latter negative result was obtained independently in 
the unpublished thesis of Bailey \cite{Bailey:76}). The positive
result admits a number of generalizations and extensions, including those to continuous alphabets 
\cite{Algoet:92,Morvai:96,Morvai:97,Gyorfi:98,BRyabko:09,BRyabko:10}.

Prediction with {\em expert advice}  (see \cite{Cesa:06} for an overview) presents a different 
approach to the problem of prediction. Here one assumes that the data source to predict is an
arbitrary deterministic sequence, and makes no further assumptions on it. The goal is 
also different: rather than trying to make the prediction error decrease to 0 (which is impossible
in this setting), it is required  to predict as well as any expert from a given set. An important  difference is that in this setting
one does not give probability forecasts of the next outcome but just deterministic predictions, 
and the quality is measured (according to some loss function) with respect to the prediction of each 
expert.   The set of experts is usually small, most typically finite; the class of all i.i.d.\ predictors
 also has been considered \cite{Freund:03}. 
While this approach is very close to the one taken here, it does not allow one to look at predictors
(experts) and data source as objects of the same kind, thus making it difficult to formulate our question of 
interest. 

 A connection between the settings was made in the work \cite{Ryabko:11pq4+}, which formulates three problems.
The first one is the classical problem of constructing a predictor that is asymptotically consistent (its error goes to 0) 
if any process from an (arbitrary, given) set $\S$ is chosen to generate the data. The second is the one considered in this work:
asymptotically consistent prediction of  sequences generated by every source for which there is an asymptotically consistent predictor
 in a given set~$\S$.
The third setting removes the  ``asymptotically consistent'' part:
it requires constructing a predictor that predicts any source whatsoever as well as any predictor in a given set $\S$.
Thus, the latter formulation is  the worst-case analysis akin to expert advice (the only difference is that 
we still try to forecast probabilities, rather than individual outcomes).
Here all predictors and sources are just time-series distributions. The three problems  are naturally ordered
in difficulty: if the set $\S$ is the same in all the three problems, then any solution to the third problem is a solution to the second, and any solution to the second
is a solution to the first. For the set of all stationary processes, it is known since \cite{BRyabko:88} 
that the first problem admits a solution. It is shown in \cite{Ryabko:11pq4+} that the third problem (worst-case) does not,
but  the question of whether the second problem admits a solution for this set was left open; here   we answer it in the negative.

\section{Preliminaries}\label{s:pre}
Let $\X$ be a finite set. Since we are after a negative result, selecting $\X:=\{0,1\}$ is not 
a restriction, so we fix this choice.  The notation $x_{1..n}$ is used for $x_1,\dots,x_n$. 
 We consider  time-series distributions, that is, probability measures on $\Omega:=(\X^\infty,\mathcal B)$ where $\mathcal B$
is the sigma-field generated by the cylinder sets  $[x_{1..n}]$, $x_i\in\X, n\in\N$ 
and $[x_{1..n}]$ is the set of all infinite sequences that start with $x_{1..n}$.
We use  $\E_\mu$ for the
expectation with respect to a measure $\mu$.

For two measures $\mu$ and $\rho$ introduce the {\em expected cumulative Kullback-Leibler divergence (KL divergence)} as
\begin{multline*}\label{eq:akl} 
  d_n(\mu,\rho)\\:=\E_\mu
  \sum_{t=1}^n  \sum_{a\in\X} \mu(x_{t}=a|x_{1..t-1}) \log \frac{\mu(x_{t}=a|x_{1..t-1})}{\rho(x_{t}=a|x_{1..t-1})}.
\end{multline*}
In words, we take the expected (over data) cumulative (over time) KL divergence between $\mu$- and $\rho$-conditional (on the past data) 
probability distributions of the next outcome.
Define also 
$$
d(\mu,\rho):=\liminf_{n\to\infty} \frac{1}{n} d_n(\mu,\rho).
$$
We say that {\em $\rho$ predicts $\mu$} (in expected average KL divergence) if 
$
d(\mu,\rho)=0.
$
It is easy to see that 
$$
d_n(\mu,\rho)=\E_\mu\log \frac{\mu(x_{1..n})}{\rho(x_{1..n})},
$$
which makes expected average KL divergence a convenient measure of prediction quality to study.
%

Let the set $\mathcal P$ be the set of all time-series distributions over $\Omega$.
A distribution $\rho\in\mathcal P$ is {\em stationary}
if for every $i,j\in\N$ and every $A\in\X^j$, 
we have $$\rho(X_{1..j} =A)=\rho(X_{i..i+j-1} =A).$$
A stationary distribution $\rho$ is called {\em ergodic} if for all $n\in\N, A\in\X^n$ 
with probability~1 we have $\lim_{n\rightarrow\infty}\nu(X_{1..n},A) = \rho(A),$ 
where $\nu(X_{1..n},A)$ stands for the frequency of occurrence of the word $A$ in $X_{1..n}$.
(The latter definition can be shown to be equivalent to the usual one formulatedf in terms of shift-invariant sets \cite{Gray:88}.)

\section{Main result}
Denote  $\S\subset \mathcal P$  the set of all stationary ergodic time-series distributions.
 Define 
$$
\Sp:=\{\mu\in\mathcal P: \exists \nu\in\S\ d(\nu,\mu)=0\}.
$$
\begin{theorem}\label{th:main} 
 For any predictor $\rho\in\mathcal P$ there is a measure $\mu\in\Sp$ such that $d(\mu,\rho)\ge 1$.
\end{theorem}
\begin{proof}
 We will show that the set $\Sp$ includes the set $\mathcal D$ of all Dirac measures, that is, 
of all measures concentrated on one deterministic sequence.  The statement of the theorem follows
directly from this, since for any $\rho$ one can find a sequence $x_1,\dots,x_n,\dots\in \X^\infty$
such that $\rho(x_n|x_{1..n-1}) \le 1/2$ for all $n\in\N$.

To show that $\mathcal D\subset\Sp$, we will construct, for any given sequence $x:=x_1,\dots,x_n,\dots\in \X^\infty$,
a measure $\mu_x$ such that $d(\delta_x,\mu_x)=0$ where $\delta_x$ is the Dirac measure concentrated on~$x$.
These measures are constructed as functions of a stationary  Markov chain with a countably infinite set of states.
The construction is based on the one used in \cite{BRyabko:88} (see also \cite{Gyorfi:98}).

The Markov chain $M$ has the set $\N$ of states. From each state $j$ it transits to the state $j+1$
with probability $p_j:={j^2}/{(j+1)^2}$ and to the state $1$ with the remaining probability, $1-p_j$.
Thus, $M$ spends most of the time around the state $1$, but takes rather long runs towards outer states: long, 
since $p_j$ tends to 1 rather fast. We need to show that it does not ``run away'' too much; more precisely, 
we need to show $M$ has a stationary distribution. For this, it is enough to show that the state $1$ is positive
recurrent (see, e.g., \cite[Chapter VIII]{Shiryaev:96} for the definitions and facts about Markov chains used
here). This can be verified directly as follows. Denote  $f_{11}^{(n)}$ the probability that starting from the state 
1 the chain returns to the state 1 for the first time in exactly $n$ steps.
 We have  
$$f_{11}^{(n)}=(1-p_n)\prod_{i=1}^{n-1}p_i=\left(1-\frac{n^2}{(n+1)^2}\right)\frac{1}{n^2}.$$
To show that the state 1 is positive recurrent we need $\left(\sum_{n=0}^\infty n f_{11}^{(n)}\right)^{-1}>0$.
Indeed, $n f_{11}^{(n)} 
<{3}/{n^2}$  which is summable.
It follows that $M$ has a stationary distribution, which we  call~$\pi$.

For   a given sequence $x:=x_1,\dots,x_n,\dots\in A^\infty$,
the measure $\mu_x$ is constructed as a function $g_x$ of the chain $M$ taken with 
its stationary distribution as the initial one. We define $g_x(j)=x_j$ for all $j\in\N$.
Since $M$ is stationary, so is~$\mu_x$. It remains to show that $d(\delta_x,\mu_x)=0$.
Indeed, we have 
\begin{multline*}
d_n(\delta_x,\mu_x)=-\log\mu_x(x_1,\dots,x_n)\le -\log\left(\pi_1\prod_{j=1}^np_j\right)
\\ = -\log\pi_1 + 2\log(n+1)=o(n).
\end{multline*}
\end{proof}
\subsection{Other sets of measures to predict and  tightness of the result}
One can ask how ``tight'' is the negative result presented, or, in other words, 
whether the set $\S$ was too general a point of departure in the first place. 

To answer this question, first note that, as mentioned before, the work \cite{BRyabko:88} shows (by an explicit construction) that there is a universal predictor for the set $\S$ (of stationary ergodic distributions) itself,
that is, there exist a measure $\rho$ such that $d(\mu,\rho)=0$ for any $\mu\in\S$.

Next, from  the proof of Theorem~\ref{th:main} one can see that it is possible to replace 
the set $\S$ in its  formulation with the set of all hidden  Markov chains with  a countably 
infinite set of states. The latter set is in fact much smaller than the set $\S$.
Indeed,  $\S$
can be considered as the set of all stationary hidden Markov processes with an uncountably infinite 
(specifically, $\X^\infty$) set of states, giving the ``much smaller'' comparison  above a precise set-theoretic meaning.

Passing to positive results for the problem of prediction considered, for the set $\mathcal M$
of all finite-memory processes, \cite[Theorem 15]{Ryabko:11pq4+} shows that there is a universal predictor 
for the set $\mathcal M^+$. Moreover, it is easy to extend the proof of the latter result
to all hidden Markov processes with finitely many states.
Thus, it is possible to predict all measures that are predicted by a hidden Markov chain 
with finitely many states, but not with a countably infinite set of states, making 
the negative result rather tight.

\subsection{Other measures of prediction quality}
So far, we have been measuring the quality of prediction in terms of expected average
KL divergence. Measuring it differently would change both the set $\Sp$ and the requirement
on the predictor that would have to predict all measures from this set. Thus, the result 
of Theorem~\ref{th:main} does not directly entail a similar statement about neither weaker nor 
stronger measures of prediction quality. 

However, a quick look at the proof of Theorem~\ref{th:main} shows that the construction 
it employs is rather universal.
 Specifically,   $\mu_x$ predicts 
$x$ also almost surely rather than in expectation (simply because the sequence $x$ is a deterministic sequence).
\newpage
Moreover, the prediction error (of $\mu_x$ on $x$) convergence to 0 in just about any sense 
one can think of, for example, one can replace KL divergence with the absolute loss, squared loss, etc.
This implies that Theorem~\ref{th:main} holds for these measures of prediction quality as well.
Furthermore, this shows that the same result holds if we consider different notions of prediction on different sides of the question:
asking whether it is possible to predict in (say) expected average KL divergence all measures that are predicted
by some stationary ergodic measure when (say) the convergence has to be with probability 1  and there is no time-averaging. 

Thus, the result (placed in the title) appears to be general and not an artefact of the measure of  prediction quality considered.

\subsection*{Acknowledgments}
Daniil Ryabko acknowledges the  support of the French Ministry of Higher Education and Research,  the Nord-Pas-de-Calais Regional Council and FEDER.
Boris Ryabko acknowledges the  support of  the Russian Foundation for Basic Research, project no. 15-07-01851


\begin{thebibliography}{10}
\providecommand{\url}[1]{#1}
\csname url@samestyle\endcsname
\providecommand{\newblock}{\relax}
\providecommand{\bibinfo}[2]{#2}
\providecommand{\BIBentrySTDinterwordspacing}{\spaceskip=0pt\relax}
\providecommand{\BIBentryALTinterwordstretchfactor}{4}
\providecommand{\BIBentryALTinterwordspacing}{\spaceskip=\fontdimen2\font plus
\BIBentryALTinterwordstretchfactor\fontdimen3\font minus
  \fontdimen4\font\relax}
\providecommand{\BIBforeignlanguage}[2]{{%
\expandafter\ifx\csname l@#1\endcsname\relax
\typeout{** WARNING: IEEEtran.bst: No hyphenation pattern has been}%
\typeout{** loaded for the language `#1'. Using the pattern for}%
\typeout{** the default language instead.}%
\else
\language=\csname l@#1\endcsname
\fi
#2}}
\providecommand{\BIBdecl}{\relax}
\BIBdecl

\bibitem{Ryabko:11pq4+}
D.~Ryabko, ``On the relation between realizable and non-realizable cases of the
  sequence prediction problem.'' \emph{Journal of Machine Learning Research},
  vol.~12, pp. 2161--2180, 2011.

\bibitem{BRyabko:88}
B.~Ryabko, ``Prediction of random sequences and universal coding,''
  \emph{Problems of Information Transmission}, vol.~24, pp. 87--96, 1988.

\bibitem{Bailey:76}
D.~H. Bailey, \emph{Sequential schemes for classifying and predicting ergodic
  processes}.\hskip 1em plus 0.5em minus 0.4em\relax Department of Mathematics,
  Stanford University., 1976.

\bibitem{Algoet:92}
P.~Algoet, ``Universal schemes for prediction, gambling and portfolio
  selection,'' \emph{The Annals of Probability}, vol.~20, no.~2, pp. 901--941,
  1992.

\bibitem{Morvai:96}
G.~Morvai, S.~Yakowitz, and L.~Gyorfi, ``Nonparametric inference for ergodic,
  stationary time series,'' \emph{Ann. Statist}, vol.~24, no.~1, pp. 370--379,
  1996.

\bibitem{Morvai:97}
G.~Morvai, S.~Yakowitz, and P.~Algoet, ``Weakly convergent nonparametric
  forecasting of stationary time series,'' \emph{Information Theory, IEEE
  Transactions on}, vol.~43, no.~2, pp. 483 --498, Mar. 1997.

\bibitem{Gyorfi:98}
L.~Gyorfi, G.~Morvai, and S.~Yakowitz, ``Limits to consistent on-line
  forecasting for ergodic time series,'' \emph{IEEE Transactions on Information
  Theory}, vol.~44, no.~2, pp. 886--892, 1998.

\bibitem{BRyabko:09}
B.~Ryabko, ``Compression-based methods for nonparametric prediction and
  estimation of some characteristics of time series,'' \emph{IEEE Transactions
  on Information Theory}, vol.~55, pp. 4309--4315, 2009.

\bibitem{BRyabko:10}
------, ``Applications of universal source coding to statistical analysis of
  time series,'' \emph{Selected Topics in Information and Coding Theory, World
  Scientific Publishing}, pp. 289--338, 2010.

\bibitem{Cesa:06}
N.~Cesa-Bianchi and G.~Lugosi, \emph{Prediction, Learning, and Games}.\hskip
  1em plus 0.5em minus 0.4em\relax Cambridge University Press, 2006.

\bibitem{Freund:03}
Y.~Freund, ``Predicting a binary sequence almost as well as the optimal biased
  coin,'' \emph{Information and Computation}, vol. 182, no.~2, pp. 73--94,
  2003.

\bibitem{Gray:88}
R.~Gray, \emph{Probability, Random Processes, and Ergodic Properties}.\hskip
  1em plus 0.5em minus 0.4em\relax Springer Verlag, 1988.

\bibitem{Shiryaev:96}
A.~N. Shiryaev, \emph{Probability}.\hskip 1em plus 0.5em minus 0.4em\relax
  Springer, 1996.

\end{thebibliography}
\end{document}